\def\notes{1}
\def\stoc{0}
\documentclass[11pt]{article}

\ifnum\stoc=0
\usepackage{amsmath, amsthm, amssymb, graphicx, enumerate, fullpage, url}
\fi
\usepackage{color}
\usepackage{mathrsfs}

\newcommand{\mnote}[1]{\ifnum\notes=1{{\sf\color{red} [Madhu: #1]}}\fi}
\newcommand{\enote}[1]{\ifnum\notes=1{{\sf\color{blue} [Elad: #1]}}\fi}
\newcommand{\fullnote}[1]{\ifnum\stoc=0 {{#1}}\fi}
\newcommand{\stocnote}[1]{\ifnum\stoc=1 {{#1}}\fi}

\newtheorem{thm}{Theorem}[section]
\newtheorem{theorem}[thm]{Theorem}
\newtheorem{definition}[thm]{Definition}
\newtheorem{proposition}[thm]{Proposition}

\newtheorem{lemma}[thm]{Lemma}

\newtheorem{claim}[thm]{Claim}

\mathchardef\mhyphen="2D
\newcommand{\A}{\mathcal{A}}
\newcommand{\B}{\mathcal{B}}
\newcommand{\C}{\mathcal{C}}

\newcommand{\elow}{E_{\rm low}}
\newcommand{\dlow}{D_{\rm low}}
\newcommand{\chain}{\mathrm{Chain}}
\newcommand{\lgt}{\mathrm{lgt}}
\newcommand{\sz}{\mathrm{sz}}
\newcommand{\col}{\mathrm{Col}}

\newcommand{\calf}{\mathcal{F}}

\newcommand{\Z}{\mathbb{Z}}
\newcommand{\eqdef}{\stackrel\triangle=}
\newcommand{\calp}{{\cal P}}

\newcommand{\calh}{{\cal H}}
\newcommand{\calu}{{\cal U}}
\newcommand{\cals}{{\cal S}}
\newcommand{\U}{\mathscr{C}}
\newcommand{\R}{\Re}
\newcommand{\E}{\mathbf{E}}
\newcommand{\argmax}{\mathop{\rm argmax}}

\newcommand{\from}{\leftarrow}

\begin{document}
\ifnum\stoc=0
\title{Deterministic Compression with Uncertain Priors}

\author{Elad Haramaty\thanks{Department of Computer Science, Technion, Haifa.
{\tt eladh@cs.technion.ac.il}. Work done in part when this author was visiting Microsoft Research New England.}
\and Madhu Sudan\thanks{Microsoft Research New England,
One Memorial Drive, Cambridge, MA 02139, USA.
 {\tt madhu@mit.edu}.}
}

\maketitle
\else

\title{Deterministic Compression with Uncertain Priors}
\numberofauthors{2}

\author{
\alignauthor
Elad Haramaty\titlenote{Work done in part when this author was visiting Microsoft Research New England.}\\
\affaddr{Department of Computer Science, Technion, Haifa.}\\
\email{eladh@cs.technion.ac.il}.
\alignauthor
Madhu Sudan\\
\affaddr{Microsoft Research New England,
One Memorial Drive, Cambridge, MA 02139, USA.}\\
\email{madhu@mit.edu}
}

\maketitle

\fi

\nocite{shannon}
\nocite{LZ77}
\nocite{CoverThomas}
\nocite{BravermanRao}
\nocite{Linial}
\nocite{ColeVishkin}

\begin{abstract}
We consider the task of compression of information when the
source of the information and the destination do not agree
on the prior, i.e., the distribution from which the information
is being generated.
This setting was considered previously by
Kalai et al. (ICS 2011) who suggested that this was a natural
model for human communication, and efficient schemes for compression
here could give insights into the behavior of natural languages.
Kalai et al. gave a compression scheme with
nearly optimal performance, assuming the source and destination
share some uniform randomness. In this work we explore the
need for this randomness, and give some non-trivial upper bounds
on the deterministic communication complexity for this problem.
In the process we introduce a new family of structured graphs
of constant fractional chromatic number whose (integral) chromatic number
turns out to be a key component in the analysis of the communication
complexity. We provide some non-trivial upper bounds on the chromatic
number of these graphs to get our upper bound, while using
lower bounds on variants of these graphs to prove lower bounds
for some  natural approaches to solve the communication complexity
question. Tight analysis of communication complexity of our problems
and the chromatic number of the underlying graphs remains open.
\end{abstract}


\section{Introduction}

The following example illustrates the questions studied
in this paper: Suppose Alice and Bob have a ranking of a set $U$ of $N$ elements,
say,
movies.
Specifically Alice's rank function is $A: [N] \to U$ and
Bob's rank function is $B: [N] \to U$ where $[N] = \{1,\ldots,N\}$ and
$A$ and $B$ are bijections with $A(i)$ naming the $i$th ranked movie
in Alice's ranking. Suppose further that Alice and Bob know that
their rankings are ``close'', specifically for every $x \in U$,
$|A^{-1}(x) - B^{-1}(x)| \leq 2$. How many bits does Alice have to send to Bob
so that Bob knows her top-ranked movie, i.e., $A(1)$? On the one
hand Bob knows $A(1)$ is one of the three element set
$S_1 = \{B(1),B(2),B(3)\}$ and so the information-content
from his point of view is bounded by $\log_2 3$ bits. Indeed
this leads to a randomized communication scheme, with Alice and Bob
sharing common randomness with $O(1)$ bits of communication. However
the deterministic communication complexity of the question is not
as easily settled. Part of the reason is that
Alice doesn't know $S_1$ and so has to ``guess'' it to communicate
$A(1)$. Still she is not clueless: She knows it is contained in
$T_2 = \{A(1),\ldots,A(5)\}$ and perhaps this can help her
communicate $A(1)$ efficiently to Bob. The question of interest
to us in this work is: Can Alice communicate $A(1)$ to Bob with
a number of bits that is independent of $N$? (Unfortunately, we do
not answer this question, though we do give a non-trivial upper bound.
We will elaborate on this later.)

The question above is a prototypical example of ``communication amid
uncertainty'', where the communicating players have fairly good information
about each other (in the example above Alice and Bob know each others
ranking of each movie to within $\pm 2$), but are not sure of each
other's information
and do not have a common-ground to base communication on.
One way to proceed
in such settings is for the players to communicate enough information
to agree on a common prior and then to use classical compression; but
this would be excessively wasteful for, say, a one-time communication.
One could hope for a direct solution which aims to establish communication
without requiring agreement on the prior, and indeed this was the
question studied by Kalai et al.~\cite{JKKS}.
Kalai et al. argue that
this models many natural forms of communication among humans where humans
are uncertain about each other's contexts, but try to communicate
efficiently despite the lack of a perfect common basis,
or without trying
to first agree on the prior. They argue in
particular that this leads to certain phenomena in natural communication
systems (natural language) that are not seen in carefully
designed communication systems (where perfect agreement on the prior
can be assumed).

The specific problem they consider is the following. Suppose
Alice wishes to communicate a message $m \in U$ to Bob, where
Alice is operating under the belief that the message is chosen
according to the probability distribution $P$ on $U$.
Bob on the other hand operates under the belief that the messages
are chosen according to a distribution $Q$ on $U$.
Both players are aware that their distributions may not be identical
but operate under the ``knowledge'' that their distributions are
close. Specifically, we say that $P$ and $Q$ are $\Delta$-close
if for all $m \in U$, we have $\log_2 P(m)/Q(m), \log_2 Q(m)/P(m) \leq
\Delta$.
(We use this to also define our distance between distributions:
The distance between $P$ and $Q$, denoted $\delta(P,Q)$, is defined
to be the minimum $\Delta$ such that $P$ and $Q$ are $\Delta$-close.)
The question Kalai et al. investigate is: What is the expected number of
bits, under distribution $P$, that Alice has to send to Bob so that
Bob can recover the message.
(We note that similar questions, in the interactive setting, were also
studied in the works of Harsha et al.~\cite{HJMR} and 
of Braverman and Rao~\cite{BravermanRao},
though their motivations were quite
different. Both works focus on the setting when sender and receiver have
different priors and are trying to generate a random
variable that is maximally correlated under their priors. In our
case the sender gets a concrete message from its prior and wishes
to communicate it. 
The focus in both works is on randomized solutions that get the 
communication complexity down to the minimum possible amount, whereas
our thrust is to use less (or no) randomness at the expense of
slightly larger communication complexity.)

Without any knowledge of $P$ and $Q$, it is still trivial for Alice
to communicate $m$ with $\log N$ bits. On the other hand, if $\Delta = 0$
(and so Alice and Bob have $P = Q$), then standard compression can
communicate this information with $H(P) + O(1)$ bits (where $H(P) =
\sum_{m \in U} P(m) \log_2 (1/P(m))$ denotes the binary entropy of $P$)
which may be much smaller than $\log N$. Kalai et al. show that
if Alice and Bob share some common random bits, then they can communicate
with each other with $H(P) + 2 \Delta + O(1)$ bits. This gives a graceful
degradation of performance when $\Delta > 0$, and indeed in many
natural instances of communication where $\Delta$ may be large (say 50),
this gives a very efficient communication mechanism amid large amounts
of uncertainty.

The assumption that Alice and Bob share a common random string is however a
major one, and is unclear how to achieve it in ``nature''.
This assumption affects the solution both technically and conceptually.
We discuss the technical implication first.
Technically,
this assumption is not made to alleviate computational complexity
considerations, but is rather to overcome a fundamental challenge.
The randomness is independent of $P$ {\em and} $Q$ and so effectively
manages to convert a solution that works for most pairs of Alice and
Bob (or rather their beliefs $P$ and $Q$) to one that works for every
pair $P$ and $Q$, with high probability over the randomness.
Unfortunately, any attempt to fix the random string leads back to a
solution that only works for most pairs of beliefs $(P,Q)$ (over any
distribution over the beliefs), but not one that works for every pair.
Thus the technical question that remains open is: ``Is there
a single solution that will work for every choice of $P$ and $Q$ with
performance roughly that of Kalai et al.?''

We now return to the conceptual implications of the assumption
of shared perfect randomness.
In terms of the motivating phenomenon of ``natural
communication among humans'', Kalai et al. suggest the presence of a
common dictionary (of say English) as presenting such shared randomness.
They do point out, however, that the assumption that such a dictionary is
a random string is mainly a convenient technical assumption, rather than
an empirically justifiable one. In particular, the assumption that
our beliefs are independent of the dictionary is not easy to justify.
Indeed the contrary may well be true: Our dictionary may well be strongly
influenced by our beliefs.
Thus one could ask - can one weaken the assumption on the shared
randomness to some much weaker notion of shared context? Our work
explores this question and gives some partial answers, while also
highlighting some intriguing communication
complexity/graph-theoretic questions that are raised by this line of work.

\subsection{Formal definitions and main results}

We start by defining the notion of an ``uncertain compression scheme''.

We let $\{0,1\}^*$ denote the set of all finite length binary strings.
For $x \in \{0,1\}^*$, let $|x|$ denote its length.
Throughout $U$, the set of all messages, will be a finite set of
size $N$.
Let $\calp(U)$ denote the space of all probability distributions over
$U$.

\begin{definition}[(Basic) Uncertain Compression Scheme]
For positive real $\Delta$
an {\em Uncertain Compression Scheme (UCS)} for distance
$\Delta$ over
the universe $U$ is given by a pair of $E:\calp(U) \times U \to
\{0,1\}^*$ and $D:\calp(U) \times \{0,1\}^* \to U$ that satisfy the
following correctness condition:
For every pair of distributions $P,Q \in \calp(U)$
that are $\Delta$-close and for every $m \in U$,
we have $D(Q,E(P,m)) = m$.
The {\em performance} of a
UCS $(E,D)$ is
given by the function
$L:\calp(U) \to \R^+$,
where $L(P) =
\E_{m \from_P U} [|E(P,m)|]$, i.e., the expected length of
the encoding under the distribution $P$.
We refer to such a scheme as a $(\Delta,L)$-UCS.
\end{definition}

In English, the definition above explicitly provides the distribution
as input to the encoding and decoding schemes, and expect the schemes
to work correctly even if the distributions used by the encoder and
decoder are not the same, as long as they are $\Delta$-close to
each other.
While in general we would like compression schemes which work for
all possible distributions $P,Q$ that are within $\Delta$ of each other,
and with no error (as expected in the definition above), some of
our schemes are weaker and work with some error, or only for some
class of distributions. We define such general UCS's below.

\begin{definition}[(General) Uncertain Compression Scheme]
For positive real $\Delta$ (for distance),
$\epsilon \in [0,1]$ (for error),
a class of distributions $\calf \subseteq \calp$,
and performance function $L:\calf \to \R^+$
a {\em $(\Delta,\epsilon,\calf,L)$-Uncertain Compression Scheme (UCS)}
over
the universe $U$ is given by a pair of $E:\calf \times U \to
\{0,1\}^* \cup \{\bot\}$ and $D:\calf \times \{0,1\}^* \cup \{\bot\}
\to U \cup \{\bot\}$ that satisfy the
following conditions:
\begin{enumerate}
\item
For every pair of distributions $P,Q \in \calf$
that are $\Delta$-close and for every $m \in U$, it is the case that
if $E(P,m) \ne \bot$ then $D(Q,E(P,m)) = m$. Furthermore $D(\bot) = \bot$.
\item
$\Pr_{m \from_P U} [E(P,m) = \bot] \leq \epsilon.$
\item
For every $P \in \calf$, we have
$\E_{m \from_P U} [|E(P,m)|] \leq L(P)$.
\end{enumerate}
\end{definition}

Note that we do not distinguish the two definitions above by name,
but rather just by the number of parameters. So if the number of
parameters is just two, then it is assumed that there is no error,
and the performance holds for all distributions.

We note that the definitions above only covers deterministic compression
schemes. A compression scheme with shared randomness can be defined
analogously, but we don't do so here.
We also stress that the choice of $P$ and $Q$ is ``worst-case''
within the family $\calf$ (as
formalized by the universal quantifier in the correctness condition).
There are no assumptions that $\calf$ is small (has only
finitely many elements), which tends to be the setting for
universal compression.
Similarly, we do not consider a sequence of messages that
need to be transmitted: Rather, we are considering one-shot
communication
with no assumptions on the distributions $P$ and $Q$, other
than that they are from $\calf$ and $\Delta$-close.

We recall that Kalai et al. present a
$(\Delta,H(P) + 2\Delta + c)$-UCS (with shared randomness)
for some constant $c \leq 3$.
We give two {\em deterministic} schemes in this paper, both having complexity
depending on $N$, but both using substantially less than
$\log N$ bits.

\begin{theorem}
\label{thm:first}
For every $\Delta \geq 0$, there
exists a $(\Delta,O(H(P) + \Delta + \log \log N))$-UCS,
i.e., a deterministic universal compression scheme that works
for all pairs $P,Q$ that are within distance $\Delta$ of each
other, and where
the expected length of encoding is at most $O(H(P) + \Delta + \log\log N)$.
\end{theorem}

The dependence on $N$ of this scheme is non-trivial
and thus may even be reasonable in ``natural circumstances''.
%
However it is not clear if such a dependency on $N$ is necessary.
Motivated by the quest to understand the dependence on $N$ more
closely, we explore schemes whose performance is not necessarily
linear in $H(P)$. Simultaneously we relax our schemes to allow
them to ``drop'' messages with $\epsilon$ probability.
We note
that if we don't do the latter, then the former is not really
a relaxation: Any error-free scheme with superlinear dependence on
$H(P)$ can be converted to one with linear dependence on $H(P)$
by a simple reduction (see Lemma~\ref{lem:linear}).

Our next theorem gives
a scheme
that is weaker than the one from Theorem~\ref{thm:first}
in its dependence on the entropy $H(P)$ and in that it
errs with non-zero probability.
But it does achieve significantly better dependence on $N$.

\begin{theorem}
\label{thm:second}
For every $\epsilon > 0$ and $\Delta \geq 0$ there exists
a $(\Delta,\epsilon, \calp(U),
\exp\left(H(p)/\epsilon+\Delta\log^{*}N\right))$-UCS,
i.e., the scheme has error probability at most $\epsilon$,
it works for all pairs of distributions $P,Q$ within distance $\Delta$
and the expected length of the encoding is at most
$\exp\left(H(p)/\epsilon+\Delta\log^{*}N\right)$.
\end{theorem}

In the above the notation $\exp(x)$ denotes a function of
the form $c^x$ for some universal constant $c$, and $\log^* N$
denotes the minimum integer $i$ such $\log^{(i)} N \leq 1$
and $\log^{(i)}$ is the logarithm function iterated $i$ times.

An alternate way to get around the barrier of Lemma~\ref{lem:linear},
which insists that schemes must have linear dependence on $H(P)$
or make some error, is to have schemes that do not work for
all possible pairs of distributions $P$ and $Q$.
As it turns out the scheme from Theorem~\ref{thm:second} does
have this behavior for many natural distributions. In
Theorem~\ref{thm:low entropy:restricted dist} we show that our scheme from Theorem~\ref{thm:second}
works without error and with same performance as long as
$P$ (or $Q$) are close to a ``flat distribution'' (uniform over
a subset), or a geometric distribution, or a binomial distribution.
We stress that the scheme is not particularly carefully tailored to
the class of distributions (though of course the encodings and decodings
do depend on the distributions), but naturally adapts to being error-free
for the above classes.

\subsection{Techniques: Graph Coloring}

While the most natural framework for studying our problem
is as a question of communication complexity of a relational
problem (as in \cite{Kushilevitz-Nisan}), this turns out
not to be the most useful for studying the deterministic communication
complexity. Indeed, as pointed out earlier, the modern stress
in communication complexity is often on designing and understanding
the limits of protocols that are interactive and use shared randomness,
while in our case the thrust in the opposite direction.

It turns out our questions are naturally also captured as graph-coloring
questions. Furthermore such questions (or related ones) have been studied in
the literature on distributed computing in the attempt color graphs in a
local distributed manner. In particular, the work of Linial~\cite{Linial}
shows that a ``local'' algorithm for 3-coloring a cycle, due to
Cole and Vishkin~\cite{ColeVishkin}, implies that a large ``high-degree
graph'' is 3-colorable. The ideas of Cole and Vishkin~\cite{ColeVishkin} and
Linial~\cite{Linial} turn out to be quite useful in our context.
Our work abstracts some of these techniques, and extends them
to get combinatorial results, which we then convert to
efficient compression schemes.

\paragraph{Uncertainty graphs and Chromatic number}

We start by defining a class of structured combinatorial graphs
whose chromatic number turns out to be central to our problems.
Let $[N] = \{1,\ldots,N\}$.
Let $S_N$ denote the set of all permutations on $N$ elements,
i.e., the set of all bijections from $[N]$ to itself.
For $\pi,\sigma \in S_N$, let
$\delta(\pi,\sigma) = \max_{i \in [N]} | \pi^{-1}(i) - \sigma^{-1}(i)|$.

\begin{definition}[Uncertainty graphs]
For integer $N,\ell$ the uncertainty graph $\calu_{N,\ell}$
has as elements of $S_N$ as its vertices, with
$\pi \leftarrow \sigma$ if (1) $\pi(1) \ne \sigma(1)$
and $\delta(\pi,\sigma) \leq \ell$.
\end{definition}

It turns out that the chromatic number of the uncertainty graphs have
a close connection to uncertain communication schemes. Roughly these
graphs emerge from a very restricted version of the communication problem,
where the distributions $P$ and $Q$ are geometric distributions (giving
probability proportional to
$\beta^{-\pi^{-1}(i)}$ and $\beta^{-\sigma^{-1}(i)}$
to the element $i \in [N]$.
It follows that if $\delta(\pi,\sigma)$ is small, then $P$ and $Q$ are
close to each other. Furthermore, for simplicity these graphs
only consider the case that the message is the element with maximal
probability under $P$. To understand how the chromatic number plays
a role, fix a receiver with distribution $Q$ and consider two possible
senders $P$ and $P'$ that could communicate with this receiver. Consider
coloring $P$ and $P'$ by  $E(P,\argmax_m\{P(m)\})$
and $E(P',\argmax_m\{P'(m)\})$ respectively. This
would lead to distinct colors on pairs $P$ and $P'$ that are too
close to each other, provided their messages, i.e., $\argmax_m\{P(m)\}$
and $\argmax_m \{P'(m)\}$ are different. This exactly corresponds to
adjacency in our graph: the underlying permutations $\pi$ and
$\sigma$ are close, and the top ranked elements are different.

The results of Kalai et al. imply that the ``fractional chromatic
number'' of $\calu_{N,\ell}$ is bounded by $O(\ell)$.\footnote{The
fractional chromatic number of a graph $G$ is the smallest positive
real $w$ such that there exists a collection of
independent sets $I_1,\ldots,I_t$ in $G$ with weights $w_1,\ldots,w_t$
such that $\sum_{j=1}^t w_j = w$ and for every vertex $u \in V(G)$ it
is the case that
$\sum_{j : I_j \ni u} w_j \geq 1$.}
The (integral) chromatic number on the other hand does not immediately
seem to be bounded as a function of $\ell$ alone. The implication of
the low fractional chromatic number is that the chromatic number of
$\calu_{N,\ell}$ is at most $O(\ell N\log N)$, but this is worse that
the naive upper bound of $N$, which can be obtained by setting the
color of $\pi$ to be $\pi^{-1}(1)$. (By definition of adjacency this is a
valid coloring.) Our main technical contribution is in obtaining
some non-trivial upper bounds on the chromatic number of this
graph.

To derive our upper bounds, we look at ``coarsened'' versions of the
graph $\calu_{N,\ell}$. For positive integer $k$, we say that
$\pi:[k] \to [N]$ is a {\em $k$-subpermutation} if $\pi$ is injective.
We let $S_{N,k}$ denoted the set of all $k$-subpermutations on $[N]$.
For $k' \geq k$, we say subpermutation $\pi:[k] \to [N]$
{\em extends} the subpermutation $\sigma:[k'] \to [N]$ if $\sigma(i) = \pi(i)$ for
all
$i \in [k]$.
For $k$-subpermutations $\pi$ and $\sigma$, we let
$\delta(\pi,\sigma) = \min_{\pi',\sigma' {\rm~extending~} \pi,\sigma}
\{\delta(\pi',\sigma')\}$.

\begin{definition}[Restricted Uncertainty graphs]
For integers $N,\ell$ and $k$ the {\em $k$-restricted uncertainty graph}
$\calu_{N,\ell,k}$
has elements of $S_{N,k}$ as its vertices, with
$\pi \leftarrow \sigma$ if (1) $\pi(1) \ne \sigma(1)$
and $\delta(\pi,\sigma) \leq \ell$.
\end{definition}

Note that $\calu_{N,\ell,N} = \calu_{N,\ell}$.
We derive our upper bounds on the chromatic number of $\calu_{N,\ell}$
by giving non-trivial upper bounds on the chromatic number of
$\calu_{N,\ell,k}$.

\begin{lemma}
\label{lem:main-color}
\begin{enumerate}
\item
For every $k \leq k'$, $\chi(\calu_{N,\ell,k'}) \leq \chi(\calu_{N,\ell,k})$.
\item
For every $N, \ell$,
$\chi(\calu_{N,\ell,2\ell}) \leq O(\ell^2 \log N)$.
\item
For every $N$, $\ell$ and $k$ that is an integral multiple of $\ell$,
we have
$\chi(\calu_{N,\ell,k}) \leq O(2^{k} \log^{(k/\ell)} N)$.
\item
For every $N$, $\ell$ and $k$ that is an integral multiple of $\ell$,
we have
$\chi(\calu_{N,\ell,k}) \geq \log^{(2k/\ell)} (N/\ell))$.
\end{enumerate}
\end{lemma}

As an immediate application we get the following theorem.

\begin{theorem}
For every $N$ and $\ell$, we have
$\chi(\calu_{N,\ell}) \leq O\left(\min\{\ell^2 \log N, 2^{\ell \log^* N}\}\right)$.
\end{theorem}

Unfortunately, the lower bound from Part (4) of Lemma~\ref{lem:main-color}
goes to $0$ as $k \to N$ and so we don't get a growing function of
$N$ as a lower bound. However, it does rule out most natural strategies
for coloring $\calu$, and shows limitations of the intuition
that suggests $\calu$ may be colorable with $f(\ell)$ colors independent
of $N$. This is so since the intuition as well most natural strategies only
use the top $O(\ell)$ ranking elements of a permutation $\pi$ to determine
its color; and such strategies are inherently limited.
In particular, it shows that there is no hope to extend the methods of Kalai et al.
in a simple way to get a deterministic UCS.

\paragraph{Organization of this paper.}
We start with the analysis of the chromatic number in
Section~\ref{sec:graph}. We then use the methods to build
uncertain compression schemes in Section~\ref{sec:ucs}.

\section{Uncertainty Graphs}
\label{sec:graph}

We start with some elementary material in Section~\ref{ssec:elem}
that already allows
us to prove Parts (1) and (2) of Lemma~\ref{lem:main-color}.
The lower bound mentioned in Part (4) of Lemma~\ref{lem:main-color}
follows also relatively easily from a result of
Linial~\cite{Linial} and we show this in
Section~\ref{ssec:graph-lb}.
Our main contribution, in Section~\ref{ssec:graph-ub}, gives the
upper bound from Part (3) of Lemma~\ref{lem:main-color}.

\subsection{Preliminaries}
\label{ssec:elem}

We recall the concept of a homomorphism of graphs:
For graph $G = (V,E)$ and $G' = (V',E')$, we say that
$\phi:V \to V'$ is a homomorphism from $G$ to $G'$
if $(u,v) \in E \Rightarrow (\phi(u),\phi(v)) \in E'$.
We say $G$ is homomorphic to $G'$ if there exists a
homomorphism from $G$ to $G'$.

\begin{proposition}
\label{prop:coarse1}
For every $N$, $\ell\geq 1$ and $k' \leq k \leq N$,
the $k$-restricted uncertainty graph $\calu_{N,\ell,k}$
is homomorphic to
the $k'$-restricted uncertainty graph $\calu_{N,\ell,k'}$.
\end{proposition}

\begin{proof}
We construct the homomorphism $\phi$ from $\calu_{N,\ell,k}$
to $\calu_{N,\ell,k'}$ as follows:
For $\pi = \langle \pi(1),\ldots,\pi(k)\rangle \in S_{N,k}$ let
$\phi(\pi) = \langle \pi(1),\ldots,\pi(k') \rangle \in S_{N,k'}$.
From the definitions it follows that this is a homomorphism.
\end{proof}

\begin{proposition}
\label{prop:coarse2}.
For every $G$ and $G'$ such that $G$ is homomorphic to
$G'$, we have $\chi(G) \leq \chi(G')$.
\end{proposition}

\begin{proof}
Follows from the composability of homomorphisms and the
fact that $G$ is $k$-colorable if and only if it is
homomorphic to $K_k$, the complete graph on $k$ vertices.
\end{proof}

Part (1) of Lemma~\ref{lem:main-color} follows immediately from
Propositions~\ref{prop:coarse1}~and~\ref{prop:coarse2}.

\begin{proposition}
\label{prop:frac-unc}
For every $N,\ell$, and $k \geq \ell+1$
the fractional chromatic number of
the restricted uncertainty graph $\calu_{N,\ell,k}$
is at most $4 \ell$.
\end{proposition}

\begin{proof}
For every function $f : [N] \to [2\ell]$ we associate
the set $I_f = \{\pi \in S_{N,k} | f(\pi(1)) = 1
{\rm~and~} f(\pi(j)) \ne 1 ~\forall ~j \in \{2,\ldots,\ell+1\}\}$.

We claim that $I_f$ is an independent set of $\calu_{N,\ell,k}$
for every $f$.
To see this consider an edge $(\pi,\sigma)$ and suppose
$\pi \in I_f$. Then $\sigma(1) \in \{\pi(2),\ldots,\pi(\ell+1)\}$
and so $f(\sigma(1)) \ne 1$ and so $\sigma \not\in I_f$.

Next we note that for every $\pi$, the probability that
$\pi \in I_f$ for $f$ chosen uniformly at random is
$1/(2\ell) \cdot (1 - 1/(2\ell))^\ell \geq 1/(4\ell)$.

Thus if we give each $I_f$ a weight of $4\ell/(2\ell)^N$,
then we have that the weight of independent sets containing
any given vertex $\pi$ is at least one, while the sum of all
weights is $4\ell$, thus yielding the claimed bound on the
fractional chromatic number.
\end{proof}

The following is a well-known connection between fractional
chromatic number and chromatic number.

\begin{proposition}
\label{prop:fractional}
For every graph $G$, $\chi(G) \leq \chi_f(G) \cdot \ln |V(G)|$.
\end{proposition}

We are now ready to prove part (2) of Lemma~\ref{lem:main-color}.

\begin{lemma}
$\chi(\calu_{N,\ell}) \leq \chi(\calu_{N,\ell,\ell+1})) \leq
4\ell(\ell+1)\ln N$
\end{lemma}

\begin{proof}
The first inequality follows from
Propositions~\ref{prop:coarse1}~and~\ref{prop:coarse2}.
The second one
follows from Proposition~\ref{prop:fractional}~and~\ref{prop:frac-unc}
and the fact that $\calu_{N,\ell,\ell+1}$
has at most $N^{\ell+1}$ vertices.
\end{proof}

\subsection{Lower Bound on Chromatic Number}
\label{ssec:graph-lb}

We now prove Part (4) of Lemma~\ref{lem:main-color} giving a lower bound
on $\chi(\calu_{N,\ell,k})$. We use a lower bound on a somewhat
related family of graphs due to Linial~\cite{Linial}.

\begin{definition}[Shift graphs]
For integers $N$ and $k < N$, we say that
$\pi \in S_{N,k}$ is a {\em left shift} of $\sigma \in
S_{N,k}$ if $\pi(i) = \sigma(i+1)$ for $i \in [k-1]$ and
$\pi(k) \ne \sigma(1)$. We say $\pi$ is a {\em right shift}
of $\sigma$ if $\sigma$ is a left shift of $\pi$, and
we say $\pi$ is a shift of $\sigma$ if $\pi$ is a left shift
or a right shift of $\sigma$.
For integers $N$ and $k$, the shift graph $\cals_{N,k}$ is
given by $V(\cals_{N,k}) = S_{N,k}$ with $(\pi,\sigma) \in
E(\cals_{N,k})$ if $\pi$ is a shift of $\sigma$.
\end{definition}

\begin{theorem}[{\protect Linial~\cite[Proof of Theorem 2.1]{Linial}}]
\label{thm:linial}
For every odd $k$, $\chi(\cals_{N,k}) \geq \log^{(k-1)} N$.
\end{theorem}

(We note that the notation in \cite{Linial} is somewhat
different: The graph $\cals_{N,k}$ is denoted $B_{N,t}$
for $t = (k-1)/2$ in \cite{Linial}.)

We show that the uncertainty graphs contain a subgraph
isomorphic to the shift graph. This gives us our lower
bound on the chromatic number of uncertainty graphs.

\begin{lemma}
For every $N$, $\ell$ and $k$ that is an integral multiple of $\ell$,
we have
$\chi(\calu_{N,\ell,k}) \geq (\log^{(\lceil 2k/\ell \rceil)} (N/\ell))$.
\end{lemma}

\begin{proof}
First without loss of generality we only consider the case
of even $\ell$. Then we reduce to the case $\ell = 2$,
by considering only those permutations $\pi$ which fix
$\pi(i) = i$ if $\ell/2$ does not divide $i$. This still
leaves us with $2N/\ell$ unfixed elements and subpermutations
from $S_{2N/\ell,2k/ell}$ that are within distance $2$ of each other
are within distance $\ell$ when mapped back to $S_{N,k}$.

So we assume $\ell = 2$ and show that $\calu_{N,2,k}$ contains
a subgraph isomorphic to the shift graph $\cals_{N,k}$.
Consider the map $\phi$ from
$V(\cals_{N,k})$ to
$V(\calu_{N,2,k})$
which send $\pi = \langle \pi(1),\ldots,\pi(k) \rangle$
to $\phi(\pi) = \sigma = \langle \sigma(1),\ldots,\sigma(k) \rangle$
as follows: Let $t = \lfloor k/2 \rfloor$. Then
$\sigma(2i) = \pi(t+i)$ and $\sigma(2i+1) = \pi(t-i)$
$\sigma(t + i) = \pi(2i)$
and $\sigma(t - i) = \pi(2i + 1)$. It is easy to verify that the map
is a bijection and if
$\pi$ and $\pi'$ are shifts of each other, then $\phi(\pi)$
and $\phi(\pi')$ are within distance $2$ of each other. It follows
that $\calu_{N,2,k}$ contains a copy of $\cals_{N,k}$ and
so $\chi(\calu_{N,2,k}) \geq \chi(\cals_{N,k}) \geq log^{(k-1)} N$.
\end{proof}

\subsection{Upper Bound on Chromatic Number}
\label{ssec:graph-ub}

In this section we give an upper bound on the chromatic number
of the uncertainty graphs. We first describe our strategy.
Fix $N$ and $\ell$.
Now for every $k$, we know that there is a homomorphism
from $\calu_{N,\ell,k}$ to $\calu_{N,\ell,k-1}$.
However we note that if we jump from $\calu_{N,\ell,k}$ to
$\calu_{N,\ell,k-\ell}$ then the homomorphism has an even
nicer property. To describe this property, we introduce
a new parameter associated with the homomorphism from
$\calu_{N,\ell,k}$ to $\calu_{N,\ell,k-\ell}$. Let us
denote this homomorphism $\phi_k$. For $\pi \in S_{N,k}$
let $d_k(\pi) = |\{\phi_k(\sigma) \mid (\pi,\sigma) \in
E(\calu_{N,\ell,k})\}|$. Note that $d_k(\pi)$ is independent
of $\pi$ and so we just denote it $d_k$. We note first that
$d_k$ is small.

Recall that $\phi_k:S_{N,k}
\to S_{N,k-\ell}$ and maps $\pi:[k] \to [N]$ to $\pi':[k-\ell]
\to [N]$ by setting $\pi'(i) = \pi(i)$.

\begin{claim}
\label{clm:dk}
For every $k$, $d_k \leq (2\ell+1)^k$.
\end{claim}

\begin{proof}
Let  $(\sigma,\pi)\in E(\calu_{N,\ell,k})$ then
$\delta(\sigma,\pi) \leq \ell$. In particular for every
$i \in [k -\ell]$, we have there exists $j(i) \in \{-\ell,\ldots,\ell\}$
such that $\sigma(i) = \pi(i + j(i))$. Thus the sequence
$j(1),\ldots,j(k-\ell)$ completely specifies $\phi_k(\sigma)$.
Since the number of such sequences is at most $(2\ell+1)^{k-\ell}$,
we get our claim.
\end{proof}

The next lemma shows that a homomorphism with a small $d$-value
yields especially good colorings.

\begin{lemma}
\label{lem:log-color}
Let $\phi$ be a homomorphism from $G$ to $H$ and let
$c = \chi(H)$ and
$d = \max_{v \in V(G)} | \{ \phi(w) \mid (v,w) \in E(G) \}|$.
Then $\chi(G) \leq 2d(d+1)\log c = O(d^2 \log c)$.
\end{lemma}

\begin{proof}
For integers $t$ and $M$,
we start by building a small family of hash functions $\calh =
\{h_1,\ldots,h_M\} \subseteq
\{ h:[c] \to [t]\}$ with the property that for
every subset $S \subseteq [c]$, with $|S| \leq d$, and for every
$i \in [c] - S$, there exists $j \in [M]$ such that
$h_j(i) \not\in \{h_j(i') | i' \in S\}$.

Given such a hash family, we claim there is a coloring of
$G$ with $t \cdot M$ colors. To get such a coloring, let
$\chi'$ be a coloring of $H$ with colors $[c]$.
Now, consider
$v \in V(G)$ and let $S_v = \{\chi'(\phi(w)) \mid (w,v) \in
E(H)\}$. By the definition of $d$, we have $|S_v| \leq d$.
Also since $\chi'$ is a coloring of $H$ and $\phi$ is
a homomorphism, we have $\chi'(\phi(v)) \not\in S_v$.
Thus by the property of $\calh$, we have that there exists a
$j = j(v)$ such that $h_j(\chi'(\phi(v)) \not\in \{h_j(i') | i' \in S_v\}$.
We let the coloring $\chi$ of $G$ be $\chi(v) = (j(v),
h_{j(v)}(\chi'(\phi(v))\}$.
Syntactically it is clear that this is a $t \cdot M$ coloring of
$G$. To see it is valid, consider $(v,w) \in E(G)$. If
$j(v) \ne j(w)$ then we are done. Else, suppose $j(v) = j(w) = j$.
Then by definition of $S_v$ we have $\chi'(\phi(w)) \in S_v$ and
so $h_j(\chi'(v)) \ne h_j(\chi'(w)) \in \{h_j(i)| i \in S_v\}$,
and thus $\chi(v) \ne \chi(w)$ as desired.

To conclude we need to give an upper bound on $t$ and $M$.

\begin{claim}
There exists such a hash family with $t \leq 2d$ and $M \leq \log (c^{d+1})$.
\end{claim}

\begin{proof}
The proof is an elementary probabilistic method argument.
Let $t = 2d$.
We pick members of $\calh$ at uniformly at random from $\{h:[c] \to [t]\}$.
Fix a set $S$ with $|S| \leq d$ and $i \in [c] - S$. Say that
$h$ separates $i$ from $S$ if $h(i) \not\in \{h(i') | i' \in S\}$.
The probability that a random $h$ separates $i$ from $S$ is at least
$1/2$ and the probability that there does not exist $h \in \calh$ separating
$i$ from $S$ is at most $2^{-M}$. The probability that there exists
$S$ and $i \in [c] - S$ such that there does not exist $h \in \calh$
separating $i$ from $S$ is strictly less than $c^{d+1} \cdot 2^{-M}$.
It follows that if $M = \log c^{d+1}$ then such a family $\calh$
exists.
\end{proof}

The lemma follows.
\end{proof}
We are now ready to prove Part (3) of Lemma~\ref{lem:main-color}, restated below.

\begin{lemma}
There exists a constant $c$ such that
for every $N,\ell,k$, we have $\chi(\calu_{N,\ell,k})
\leq 2^{c k \log \ell} \log^{(\lfloor (k-1)/\ell \rfloor-1)} N$.
\end{lemma}

\begin{proof}
We prove the lemma by induction on $k$. For notational
simplicity assume $k-1$ is a multiple of $\ell$.
For $k \leq \ell$
the lemma is immediate from the fact that
$\chi(\calu_{N,\ell,1}) \leq N$.
Assume the lemma is true for $k - \ell$.
Then,
by Lemma~\ref{lem:log-color} we have that
for $\chi(\calu_{N,\ell,k}) \leq 2 d_k (d_k+1) \cdot \log
(\chi(\calu_{N,\ell,k-\ell}))
\leq 4 d_k^2 \log \chi (\calu_{N,\ell,k-\ell})$.
By Claim~\ref{clm:dk}, $d_k \leq (2\ell+1)^k \leq (4\ell)^k$ and so
for $\chi(\calu_{N,\ell,k}) \leq 4 (4\ell)^{2k} \log (2^{c (k-\ell)\log \ell}\log^{(k-\ell-1)/ell} N \leq 2^{ck\log \ell} \log^{(k-1)/\ell)}
N$ for a suitably large $c$.
\end{proof}

\section{Uncertain Communication}

\label{sec:ucs}

We now convert some of the methods from the previous section
into schemes for uncertain compression.
In Section~\ref{ssec:simple-ucs} we derive a simple compression scheme based
on the relationship between fractional chromatic number and
chromatic number from Section~\ref{ssec:elem}.
We then use the ``nested series of homomorphisms'' from Section~\ref{ssec:graph-ub}
to derive a second compression scheme in Section~\ref{ssec:logstar-ucs}.
The compression scheme
of Section~\ref{ssec:logstar-ucs} can make errors with positive
probability and has a non-linear dependence on entropy.
In Section~\ref{ssec:natural} we show that for some
natural distributions, this scheme is error-free.
In Section~\ref{ssec:linear} we show how an error-free scheme
working for all distributions would automatically have
linear dependence on the entropy, suggesting some of the
weaknesses in Section~\ref{ssec:logstar-ucs} are necessary.

\subsection{A simple, zero-error compression scheme}
\label{ssec:simple-ucs}

Our first construction uses the notion of an isolating
hash family.
For positive integers $\ell$, $N$ and $m \in [N]$
and $S \subseteq [N] - \{m\}$, we say that a
function $h: [N] \to \{0,1\}^\ell$ isolates $m$ from
$S$ if $h(m) \not\in \{h(m') | m' \in S\}$.
We say that a hash
family $\calh_\ell =
\{h_{1,\ell},\ldots,h_{M,\ell}\}$ is {\em $(N,\ell)$-isolating}
if
for every $S \subseteq [N]$ with $|S| \leq 2^{\ell-1}$, and for
every $m \in [N] - S$, there exists $j = j(m,S)$ such
that $h_{j,\ell}(m) \not\in h_{j,\ell}(S) \eqdef \{h_{j,\ell}(m')
| m' \in S\}$.

We note first that small isolating families exist and then give a
compression scheme based on small isolating families.

\begin{lemma}
\label{lem:isol}
For every $\ell$ and $N$, there exists an
$(N,\ell)$-isolating family of size at most $2^\ell \cdot \log N$.
\end{lemma}

\begin{proof}
The proof is straightforward application of the probabilistic method.
We pick $\calh = \{h_1,\ldots,h_M\}$ by picking $h_i$ uniformly and
independently from the set of all functions from $[N]$ to
$\{0,1\}^\ell$.
Fix $m \not\in S \subseteq [N]$.
The probability that a randomly chosen $h$ isolates $m$ from $S$
is at least $1/2$.
Thus the probability that some $h_i$ in $\calh$ does not isolate
$m$ from $S$ is at most $2^{-M}$.
Taking the union bound over all $m, S$ we find that the probability
that $\calh$ does not isolate some $m$ from $S$ is at most
$N^{2^{\ell}}/2^{M}$. We conclude that
$M \leq 2^{\ell} \cdot \log N$ suffices for the existence of
such a $\calh$.
\end{proof}

We are now ready to describe
our encoding and decoding schemes.

\noindent {\bf Encoding:} Given $m,P$
let $S = \{m' \in [N] \setminus \{m\} \mid P(m') \geq P(m)/2^{2\Delta}\}$
and let $\ell = \log_2 1/P(m) + 2\Delta$.
Let $\calh$ be an $(N,\ell)$-isolating family of size $M$
and let $\calh = \{h_{1,\ell},\ldots,h_{M,\ell}\}$.
Now let $j \in [M]$ be such
that $h_{j,\ell}(m) \not\in \{h_{j,\ell}(m') \mid m' \in S\}$.
The encoding $E(P,m)$ is defined to be $(j,h_{j,\ell}(m))$.

\medskip

\noindent {\bf Decoding:} Given $Q$ and $y = (j,z) \in \Z^+ \times
\{0,1\}^*$,
let $\ell = |z|$ and let $\hat{m} = \argmax_{m \in [N] : h_{j,\ell}(m) = z}
\{Q(m)\}$. The decoding of the pair $Q,y)$ is given by $D(Q,y) = \hat{m}$.

Our next proposition verifies the correctness of the
compression scheme.

\begin{proposition}
\label{prop:simple-correct}
For every pair of distributions $P$, $Q$ such that
$\delta(P,Q) \leq \Delta$, and for every message $m \in [N]$,
it is the case that $D(Q,E(P,m)) = m$.
\end{proposition}

\begin{proof}
Fix
$P$, $Q$ and $m$ such that $\delta(P,Q) \leq
\Delta$.
Let $E(m,P) = (j,z)$ with $\ell = |z|$ and let
$D((j,z),Q) = \hat{m}$. We will show that $\hat{m} = m$.
By definition of $E$, we have $h_{j,\ell}(m) = z$
and by definition of $D$ we have $h_{j,\ell}(\hat{m}) = z$.
Thus, by the condition that $\hat{m}$ maximizes probability under
$Q$ of messages satisfying $h_{j,\ell}(m') = z$, we have
$Q(\hat{m}) \geq Q(m)$. Since the distance of $P$ and
$Q$ is at most $\Delta$, we have $P(m) \leq Q(m) 2^{\Delta}$
and $P(\hat{m}) \geq Q(\hat{m})/2^{\Delta}$. Combining the
inequalities we get $P(\hat{m}) \geq P(m)/2^{2\Delta}$.
Now let $S = \{m' \in [N] - \{m\} \mid
P(m') \geq P(m)/2^{2\Delta}\}$.
We have $\hat{m} \in S \cup \{m\}$. But by definition of
$j$, we have $h_{j,\ell}(m) \not\in \{h_{j,\ell}(m') | m' \in S\}$
and since $h_{j,\ell}(m) = h_{j,\ell}(\hat{m})$, we must
have $m = \hat{m}$.
\end{proof}

Finally we analyze the performance of our scheme.

\begin{lemma}
\label{lem:simple-anal}
The expected length of the encoding $E$ is
$O(H(P) + \Delta + \log\log N)$.
\end{lemma}

\begin{proof}
Fix $m \in S$. Then we have $\ell \leq  1 + \log 1/P(m) + 2\Delta$ and
$M \leq (2^\ell \log N)$.
Thus, the length of $E(P,m)$ is at most $2\ell + \log \log N
= O(\log 1/P(m) + \Delta + \log \log N)$.
Taking expectation over $m$ drawn from $P$, we have
the expected length of the encoding is at most $O(H(P) + \Delta + \log \log
N)$.
\end{proof}

Theorem~\ref{thm:first} follows immediately from Proposition~\ref{prop:simple-correct} and Lemma~\ref{lem:simple-anal}.

\subsection{Compression with error in the low entropy setting}
\label{ssec:logstar-ucs}
Our compression for the low entropy setting (with better dependence
on $N$) relies on an extension of our coloring scheme for the
uncertainty graphs. We describe this extension in the next
section and then use that to present our compression scheme
afterwards.

\subsubsection{Compression for chains}

We start with some terminology.
We say that a finite sequence of sets
$A_0,\ldots,A_k$ with $A_i \subseteq
N$ is a {\em chain} in $[N]$ if $|A_0| = 1$ and $A_i \subseteq A_{i+1}$
for every $i$. We say that $w$ is the {\em leader} of the
chain if $A_0 = \{w\}$. We use $\chain(N)$ to denote the
set of all chains in $[N]$.

In this section we will show how to compress the leader of a chain
so that it is unambiguous relative to ``nearby'' chains.
This is in the spirit of the coloring of uncertainty graphs. Indeed
vertices of the uncertainty graph $\calu_{N,\ell,k}$
correspond to chains with the vertex $\langle \pi(1),\ldots,\pi(k)\rangle$
corresponding to the chain $\A$ with $A_0 = \{\pi(1)\}$ and
$A_i = \{\pi(1),\ldots,\pi(\ell \cdot  i)\}$ for $i \geq 1$.
The compressing scheme will thus be similar to the coloring
scheme, however there are two distinguishing factors: We will
want to compress some chains more than others - a notion that would
correspond to asking some vertices to use small colors while
allowing others to use larger ones. Furthermore our chains will
now grow arbitrarily fast (and not just in steps of $1$ or more generally
$\ell$).
We now describe the precise problem.

For a chain $\A = \langle A_0,\ldots,A_k \rangle$ we say
the length of the chain, denoted $\lgt(\A)$, is the parameter $k$.
We use $\sz(\A)$ denote the size of the final set $|A_k|$.
For a chain $\A$ of length at least $i$, we let $\A_i$ denote
its prefix of length $i$, i.e., $\A_i = \langle A_0,\ldots,A_i\rangle$.

For chain $\A = \langle A_0,\ldots,A_k \rangle$ and
chain $\B=\langle B_0,\ldots,B_{k-d} \rangle$, we say
$\B$ is within distance $d$ from $\A$ if for all $i\in\{0,...,k-d\}$,
$A_{i-d}\subseteq B_{i}\subseteq A_{i+d}$
(where we consider sets with negative index to be the empty set).
We denote the
set of all chains that are within $d$ distance from $\A$ by $S^{d}(\A)$.
Our goal next is to compress the leader of
chains so that the length of the
compression is small as a function of $\sz(\A)$, while it remains
unambiguous to chains that are nearby.

\begin{lemma} \label{lem: chain coloring scheme}
There exists a coloring scheme $\col:\Z^+ \times \chain(N)  \to \Z^+$ with
the following properties:
\begin{enumerate}
\item If $\lgt(\A) \geq 2k$, then for every $s\geq \sz(\A_{2k})$,
$\col(s,\A_{2k}) \leq 2^{6(s+1)}\log^{(k)}N$.
\item Let $\A$ and $\A'$ be chains of the same
length, with $\lgt(\A) \geq 2k$ and of size at most $s$. Then, if $S^{1}(\A) \cap
S^{1}(\A') \ne \emptyset$ and $A_0 \ne A'_0$, then
$\col(s,\A_{2k}) \ne \col(s,\A'_{2k})$.
\end{enumerate}
\end{lemma}

\begin{proof}
Let $c_{k,s} = 2^{6(s+1)}\log^{(k)}N$.
Fix $s \geq \sz(\A)$.
We now describe a coloring scheme of a chain $\A_{2k}$
with $c_{k,s}$ colors, using induction on $k$.

For the base case $k=0$, 
Let $w$ be the leader of $\A$. Then $\A_0$ gets
the color $\col(s,\A_0) = w$, so clearly $\col(s,\A_0)
\leq N = \log^{(0)} N$.

For $k\geq 1$, let
$\ell=2.5s$ and let
$\calh$ be an $(\ell,c_{k-1,s})$-isolating family
(where isolating families were defined as in
Section~\ref{ssec:simple-ucs}).
By Lemma~\ref{lem:isol} such a family of size $M = 2^{\ell}\log c_{k-1,s}$
exists, so let $\calh = \left\{ h_{i}\right\} _{i=1}^{M}$.
Let $T = \{\B | \lgt(\B) = 2k-2, \B \in S^{2}(\A_{2k}), \col(s,\B)
\ne \col(s,\A_{2k-2})\}$.
Let $j\in[2^{\ell}\log c_{k-1,s}]$
be such that $h_{j}\left(\col(s,\A_{2k-2})\right)
\neq h_{j}\left(\col(s,\B)\right)$
for all $\B \in T$.
With these definitions in place, we define $\col(s,\A_{2k})$ to be $\left(j,h_{j}\left(\col(s,\A_{2k-2})\right)\right)$.
We verify below that this is a ``small'' coloring and a valid one.

Let us identify the set $[2^{\ell}\log c_{k-1,s}]\times\{0,1\}^{\ell}$ with
$\left[2^{2\ell}\log c_{k-1,s}\right]$.
The bound on $c_{k,s}$ follows from the fact that
\begin{eqnarray*}
\lefteqn{2^{2\ell}\log c_{k-1,s}}\\
& \leq & 2^{5 s}\log\left(2^{6(s+1)}\log^{(k-1)}N\right)\\
 & \leq & 2^{5s}\left(6(s+1)+\log^{(k)}N\right)\\
 & \leq & 2^{6(s+1)}\log^{(k)}N,
\end{eqnarray*}
where the final inequality follows from the fact that $2^s\cdot 2^6 \geq 6(s+1)$
which is true for every $s \geq 0$.

We now verify that the coloring satisfies the requirement in Part (2) of the lemma
statement, i.e., that for chains $\A$ and $\A'$ of the
same length and size at most $s$, if their prefixes have the same colors,
then they
have the same leader. Again we proceed by induction on $k$.
Assume $\col(s,\A_{2k}) = \col(s,\A'_{2k})$.

For $k=0$, by assumption we have $\col(s,\A_{0})=\col(s,\A'_{0})$.
But by definition $\col(s,\A_0) = w$ where $w$ is the leader
of $\A_0$. It follows thus that $w$ is also the leader of $\A'_0$
as claimed.

Now consider $k\geq 1$.
Let $\col(s,\A_{2k}) = (j,h_j(\col(s,\A_{2k-2})))$
and
$\col(s,\A'_{2k}) = (j',h_{j'}(\col(s,\A'_{2k-2})))$.
Since $\col(s,\A'_{2k})=\col(s,\A_{2k})$,
we have $j = j'$. Moreover, $h_j(\col(s,\A_{2k-2})) = h_j(\col(s,\A'_{2k-2}))$.

We now show that  $\A'_{2k-2} \in S^2 (\A_{2k})$. Let $\B \in S^1(\A) \cap S^1(\A')$ and consider its prefix $\langle B_0,\ldots,B_{2k-1}\rangle$. So, for every $i\in\{0,...,2k-1\}$
$$A_{i-1}\subseteq B_{i}\subseteq A_{i+1} \; \mbox{and}\; A'_{i-1}\subseteq B_{i}\subseteq A'_{i+1}\;.$$
In other words, for all $i\in\{0,...,2k-2\}$
$$A_{i-2}\subseteq B_{i-1}\subseteq A'_{i}\subseteq B_{i+1}\subseteq A_{i+2}\;,$$
Hence $\A_{2k-2}'\in S^{2}(\A_{2k})$.

From our choice of $j$, $h_j(\col(s,\A_{2k-2})) = h_j(\col(s,\A'_{2k-2}))$ for $\A'_{2k-2} \in S^2 (\A_{2k})$
only if $\col(s,\A'_{2k-2}) = \col(s,\A_{2k-2})$.
For conclusion, $\A_{2k-2}$ and $\A'_{2k-2}$ are both chains of size at most $s$ of the same length, and have the same color. From the induction hypothesis  they have the same leader.
\end{proof}

\subsubsection{The Compression Scheme}

We are now ready to define our final compression scheme.

\textbf{Encoding}:
Given $m,P$ define
$r=\left\lfloor - \log P(m)\right\rfloor $
and $f=2\left\lfloor \log^{*}N\right\rfloor -1$.
Further define the chain $\A$ of length $f$
as follows. $A_{0}=\{m\}$
and $A_{k}
=\left\{ m' \in [N] \mid
|\log 1/P(m') - r | \leq \Delta+1
\right\}$
(so that $A_k$ is the set of messages of probability roughly $P(m)$
with the difference in logarithms being at most $(k+1)\Delta+1$).
Let $s = \sz(\A)$.
The encoding $\elow(P,m) = E(P,m)$ is
\[
E(P,m)=\begin{cases}
\left(s,r,\col\left(s,\A\right)\right) & \mathrm{if~}s \leq2^{\frac{H(P)}{\epsilon}+2\Delta\log^{*}N+1}\\
\bot & \mathrm{otherwise.}
\end{cases}
\]
(We assume that $s$ and $r$ above are encoded in some prefix-free
encoding, so that the receiver can separate the three parts.)

\textbf{Decoding}: The decoding function $\dlow(Q,y) = D(Q,y)$
works as follows: If $y = \bot$ then the decoder
outputs $\bot$. Else let $y=(s,r,c)$ and let
$f = 2\lfloor \log^* N \rfloor - 1$.
Let $\B=\langle B_0,\ldots,B_{f-1} \rangle$ be as follows:
$B_{0}=\{w\}$ such
that $| \log 1/Q(w) - r | \leq \Delta + 1$.
For $k \geq 1$,
$B_{k}=\left\{ m'\mid | \log 1/Q(m/) - r | \leq (k+1)\Delta + 1
\right\}$.
Find a chain $\A'$ with the following properties: $ \B \in S^{1}(\A')$, $\lgt(\A') = f$, $\sz(\A') \leq s$  and
$\col(s,\A')=c$. Let $\hat{m}$ be the leader of $\A'$. The decoding $D(Q,y)$ is set
to be $\hat{m}$.

We first analyze the correctness of the decoder.

\begin{lemma}\label{low entropy: correctness}
For every pair of distributions $P$, $Q$ such that
$\delta(P,Q)\leq\Delta$ and for every message $m\in[N]$
such that $\elow(P,m)\neq\bot$, it
holds that $\dlow(Q, \elow(P,m))=m$.
\end{lemma}

\begin{proof}
Fix $P\in\calp([N])$ and a message $m\in[N]$ such
that $\elow(P,m)\neq\bot$. The following claims will show
that the decoding process is well defined (and then correctness
will be essentially be immediate).

\begin{claim} \label{cla:m in B_1}
There exists $w\in[N]$ such that $| \log 1/Q(w) - r | \leq
\Delta + 1$.
\end{claim}

\begin{proof}
By our choice of $r$, we have $|\log 1/P(m) - r | \leq 1$.
Now using
$\delta(P,Q)\leq\Delta$, we have $|\log 1/P(m) - \log 1/Q(m) | \leq
\Delta$, and so $|\log 1/Q(m) - r| \leq \Delta+1$.
So $w=m$ gives an element in $[N]$ with the
desired property.
\end{proof}

Thus the chain $\B$ is now well-defined. It remains to show that
there exists a chain $\A'$ satisfying the required properties.
The next claim shows that $\B \in S^1(\A)$, therefore $\A$ is
a candidate for the role of $\A'$.
\begin{claim}
\label{A in S(B)}
$\B \in S^1(\A)$ .
\end{claim}
\begin{proof}
The proof follows easily from our choice of $\A$, $\B$ and the fact that $P$ and $Q$ are $\Delta$-close.
Let $k\in\{0,...,f-1\}$. We need to show that $B_{k}$ is sandwiched
between $A_{k-1}$ and $A_{k+1}$.

First, We will show that $B_{k}\subseteq A_{k+1}$.
When $k=0$, we need to show that $w\in A_{1}$. Indeed,
\begin{eqnarray*}
& & |\log 1/Q(w)  - r| \leq \Delta + 1 \\
& \Rightarrow & | \log 1/P(w) - r | \leq 2\Delta + 1\\
 & \Rightarrow & w\in A_{1}\;.
\end{eqnarray*}

Now consider $1\leq k\leq f-1$. We have,
\begin{eqnarray*}
B_{k} & = & \left\{ m'\in[N]\mid
| \log 1/Q(m') - r | \leq (k+1)\Delta + 1
\right\} \\
& \subseteq & \left\{ m'\in[N]\mid
|\log 1/P(m') - r | \leq (k+2)\Delta + 1
\right\} \\
 &=&A_{k+1}\;.
\end{eqnarray*}
This shows that $B_{k}\subseteq A_{k+1}$.
Next we show that $A_{k-1}\subseteq B_{k}$,
for $2\leq k\leq f-1$. We have
\begin{eqnarray*}
A_{k-1} & = & \left\{ m'\in[N]\mid
|\log 1/P(m') - r| \leq k\Delta+1
\right\} \\
& \subseteq & \left\{ m'\in[N]\mid
|\log 1/Q(m') - r| \leq (k+1)\Delta + 1
\right\}\\
 & =& B_{k}\;.
\end{eqnarray*}
The case where $k=1$ and $w\in B_{1}$ was proved in Claim~\ref{cla:m in B_1}. So we are done.
\end{proof}

To conclude, the decoder  can find a chain $\A'$ such that
$\sz(A') \leq s$, $\lgt(\A') = \lgt(\A)$, $\col(s,\A')=\col(s,\A)$
and there exists a chain $\B \in S^1(\A') \cap S^1(\A)$.
From Lemma~\ref{lem: chain coloring scheme} the leader
of $A'$ is $m$ as required.
\end{proof}

We are now ready to prove Theorem~\ref{thm:second}.

\begin{proof}
We now estimate the probability that the encoder will fail. Fix some
probability $P$ and a message $m$ such that $E(P,m)=\bot$. We will first show that $P(m) \leq 2^{-\frac{H(P)}{\epsilon}}$. Later,  we will bound the probability that ``$m$ has such small probability'' by $\epsilon$.

Consider the chain $\A = \langle A_0,\ldots,A_{f} \rangle$ as defined by the encoder.
In this case, the size of the largest set, $|A_{f}|$, is more then the threshold $T=2^{\frac{H(P)}{\epsilon}+2\Delta\log^{*}N+1}$.
So, there is some element $m'\in A_{f}$ such that $P(m')\leq\frac{1}{T}$.
By our choice of $A_{f}$, $P(m')\geq 2^{-\left\lfloor -\log P(m)\right\rfloor -(f+1)\Delta-1}\geq P(m) 2^{-2  \Delta \log^{*}N-1}$.
Calculating,
\[
\frac{1}{T}\geq P(m) 2^{-2  \Delta \log^{*}N-1}\Rightarrow P(m)\leq\frac{2^{2  \Delta \log^{*}N+1}}{T}=2^{-\frac{H(P)}{\epsilon}}
\]

Therefore, we can bound the failure probability by the probability that
$P(m)\leq2^{-\frac{H(P)}{\epsilon}}$. Using the fact that
$\E_{m\leftarrow_{P}[N]}\left[\log\frac{1}{P(m)}\right] = H(P)$,
we deduce the following by Markov's inequality,
\[
\Pr_{m\leftarrow_{P}[N]}\left[P(m)\leq2^{-\frac{H(P)}{\epsilon}}\right]=\Pr_{m\leftarrow_{P}[N]}\left[\log\frac{1}{P(m)}\geq\frac{H(P)}{\epsilon}\right]\leq\epsilon
\]

We will finish the proof by bounding the performance of the scheme.
To this end
consider a distribution $P$ and a message $m\in[N]$ such that $E(P,m)\neq\bot$
(i.e $ \sz(\A) \leq T)$. The encoder sends $r=\left\lfloor -\log P(m)\right\rfloor $,  $s=\sz(\A)$
and $\col(s,\A)$. 
We first analyze the contribution of sending $r$ to the performance.
Because $\log|r| = O \left( \log(\frac{1}{P(m)} ) \right)$, the accepted
length of sending $r$ in a prefix-free encoding is at most $O\left( \E_{m\leftarrow_{P}[N]}\log(\frac{1}{P(m)})\right) = O \left(H(P)\right)$.

Now we analyze the length of $\left(s, \col(s,\A) \right)$ . By Lemma~\ref{lem: chain coloring scheme}:
\[
C(s,\A)\leq2^{6(s+1)}\log^{(f)}N=2^{O\left(s\right)}
\]
Hence, the length of $\left(s, \col(s,\A) \right)$ is at most
 $$O(\log s) + \log C(s,\A)=O(s)=2^{\frac{H(P)}{\epsilon}+2\Delta\log^{*}n+O(1)}\;.$$
Thus,
from the  linearity of expectations,
it follow that the total performance is at most
$2^{\frac{H(P)}{\epsilon}+2\Delta\log^{*}n+O(1)}$.
\end{proof}

\subsection{Error-free Compression for Natural Distributions}
\label{ssec:natural}

In this section we will show that for a large class of natural distributions, the
above scheme is error free. We start by describing the natural
distributions we can capture.

We say that a distribution $P\in \calp([N])$ is {\it flat} it there exists a set $S\subseteq [N]$ such that $P$ is uniform on $S$.
The distribution is  called {\it geometric} if there exists parameter $\alpha \in (0,1)$ and a permutation $\pi$ on $[N]$ such that for all 
$k \in [N-1]$ it holds that  $P(\pi(k+1)) = \alpha P(\pi(k))$.
We  call $P$ {\it binomial} if there exists a parameter $p \in (0,1)$ and
a permutation $\pi$ on $[N]$ such that $\forall k \in [N]$, 
$P(\pi(k)) = {{N}\choose{k}} p^{k} {(1-p)}^{n-k}$.
The sets of all flat, geometric and binomial distributions over $[N]$ are denoted by $\mathrm{Flat}_N$,  $\mathrm{Geo}_N$ and  $\mathrm{Bin}_N$ respectively.

The following theorem shows that the scheme $(\elow,\dlow)$ performs well {\em without error} on all of the above natural distributions.
Moreover, this theorem is stable in the sense that the guarantee on 
the performance holds even if a distribution is only close to one of 
the above-mentioned natural distributions.

\begin{thm}
\label{thm:low entropy:restricted dist}
Let $\calf \triangleq \mathrm{Flat}_N \cup \mathrm{Geo}_N \cup \mathrm{Bin}_N$
and $L(P) \triangleq 2^{H(P)}\left\lceil \Delta\log^{*}N\right\rceil$. Then
the scheme $(\elow,\dlow)$ (with $\epsilon$ set to $0$) is
a $\left(\Delta,0, \calf,O\left(L(P)\right)\right)$-UCS.
Moreover, if $P\in \calp([N])$ is $\Delta \log^{*} N$-close to a distribution $\tilde{P} \in \calf$
then the performance of the scheme on $P$ is
$\E_{m \from_P U} [|E(P,m)|] = O\left(L(\tilde{P})\right)$.
\end{thm}

We prove the theorem above by identifying a broad condition on
distributions, which we call the {\em capacity}, and showing that
the performance of our scheme is good if the capacity is small.
We define this notion next, show that it is small for the distributions
under consideration in Lemma~\ref{lem:natural dist}
next, and finally bound the 
performance as a function of the capacity in Lemma~\ref{lem:perf-capacity}
afterwards,
thus leading to a proof of Theorem~\ref{thm:low entropy:restricted dist}.

Let $P\in\calp([N])$ be a distribution and
let $S\subseteq[N]$ be its support. We
say that $U\subseteq S$ is a {\it unit set} of $P$ if for any two
elements $m_{1},m_{2}\in U$ the distance $\left|\log P(m_{1})-\log P(m_{2})\right|\leq1$.
We define the {\it capacity}  of $P$, denoted by $\U(P)$, to be
the minimal $c\in \R$ such that the size of
every unit set of $P$ is bounded by $2^{c}$.

\stocnote{The next lemma will show that, for the previously discussed distributions,
the capacity is roughly  the entropy. We omit the proof from this version.}
\fullnote{Later, we will prove the next lemma, showing that for the previously discussed distributions,
the capacity is roughly  the entropy.
}

\begin{lemma}\label{lem:natural dist}
Let $P \in \mathrm{Flat}_N \cup \mathrm{Geo}_N \cup \mathrm{Bin}_N$. Then $\U(P)\leq H(P)+O(1)$.
\end{lemma}

Theorem \ref{thm:low entropy:restricted dist} follows immediately
from Lemma~\ref{lem:natural dist} combined with the following lemma.

\begin{lemma}
\label{lem:perf-capacity}
For every $P$
$(\elow,\dlow)$ (with respect to $\epsilon=0$)
is a $\left(\Delta,O\left(\log\left(H(P)\right)+2^{\U(P)}\left\lceil \Delta\log^{*}N\right\rceil \right)\right)$
scheme. Moreover, if $P$ is $\Delta\log^{*}N$ close to a distribution
$\tilde{P}$, then the performance of the scheme on
$P$ is $O\left(\log\left(H(P)\right)+2^{\U(\tilde{P})}\left\lceil \Delta\log^{*}N\right\rceil \right)$.
\end{lemma}
\begin{proof}
When setting $\epsilon=0$, the encoder never outputs $\bot$.
Lemma~\ref{low entropy: correctness} already implies the correctness of the scheme.
The only remaining task is to analyze the performance of the scheme.

Recall, the output of the encoder has three components: $r$, $s$ and  $C(s,\A)$.
From linearity of expectation it suffices to analyze the expected length
of each component separately.

For a given word $m\in [N]$,
the first component
is $ $$r=\left\lfloor \log\frac{1}{P(m)}\right\rfloor$.
Its length is $|r|=O(\log\log\frac{1}{P(m)})$. Using the concavity
of the function $\log$ we can bound the expectation of $|r|$ as
follows:
\[
\E\left[\log\log\frac{1}{P(m)}\right]\leq\log\left(\E\left[\log\frac{1}{P(m)}\right]\right)=\log\left(H(P)\right)\;.
\]

Now consider the chain $\A$ with size $s$ and length $f=\log^{*}(N)-O(1)$
as define by the encoder.
The second component is the size $s$.
Clearly,
$|s|=O(\log s)$.

The third component is $C(s,\A)$. By Lemma~\ref{lem: chain coloring scheme},
 $C(s,\A) = \exp(s)$ so $|C(s,\A)|=O(s)$.

Hence the expected length of the last two components is bounded by $O(s)$. Let
$\tilde{P}\in\calp([N])$ be a distribution that is $\Delta \log^{*}N$-close to $P$.
To achieve the results it is enough
to show that the size $s$ of
the chain $\A$ associated with $P$ and $m$ is bounded by $O\left(2^{\U(\tilde{P})}\left\lceil \Delta\log^{*}N\right\rceil \right)$.

The size $s=\sz(\A)$ is the size of the following set,
\[
A=\left\{ m'\in[N]\mid\left|\log 1/P(m') - r \right|\leq2\left\lfloor \Delta\log^{*}N\right\rfloor +1\right\}
\;.
\]
We will show that this set can be covered by
$O(\left\lceil \Delta\log^{*}N\right\rceil )$
unit sets of $\tilde{P}$. This will yield an upper bound on $s$
of $O\left(2^{\U(\tilde{P})}\left\lceil \Delta\log^{*}N\right\rceil
\right)$
as required.

Let $k=3\left\lceil \Delta\log^{*}N\right\rceil +1$. Define $U_{-k},....,U_{k-1}$
as
\[
U_{i}=\left\{ m'\mid i\leq r+\log \tilde{P}(m')\leq i+1\right\} \;.
\]
Clearly the $U_{i}$s are unit sets of $\tilde{P}$. Moreover, their union is
the set
\[
\bigcup_{i=-k}^{k-1}U_{i}=\left\{ m'\mid\left|\log 1/\tilde{P}(m') - r \right|\leq3\left\lceil \Delta\log^{*}N\right\rceil +1\right\} \;.
\]
Let $m'\in A$. It remains to verify that $m'\in\bigcup_{i=-k}^{k-1}U_{i}$. Indeed,
\begin{eqnarray*}
\left|\log 1/\tilde{P}(m') - r\right| & \leq & \left|1/\log P(m') - r\right|+\Delta\log^{*}N\leq\\
 & \leq & 3\left\lceil \Delta\log^{*}N\right\rceil +1\;.
\end{eqnarray*}
Therefore, $|A| \leq \sum |U_i | =  O\left(2^{\U(\tilde{P})}\left\lceil \Delta\log^{*}N\right\rceil \right)$ as required.
\end{proof}

\fullnote{
To complete the proof of Theorem~\ref{thm:low entropy:restricted dist}, we will prove Lemma~\ref{lem:natural dist}.
The proof follows immediately from the next three claims.

\begin{claim}
Let $P\in\mathrm{Flat}_{N}$. Then $\U(P)\leq H(P)$.\end{claim}
\begin{proof}
Let $S\subseteq[N]$ be the support of $P$. Clearly, $H(P)=\log|S|$.
For every $U\subseteq S$ that is a unit set of $P$,
\[
|U|\leq|S|=2^{H(P)}\;.
\]
Thus, $\U(P)\leq H(P)$.\end{proof}

\begin{claim}\label{cla:natural dist: Geometric}
Let $P\in\mathrm{Geo}_{N}$. Then $\U(P)\leq H(P)+O(1)$.\end{claim}
\begin{proof}
Let $\alpha\in(0,1)$ be such that for all $k\in[N-1]$, $P(k+1)=\alpha P(k)$.
We will assume that $\alpha^N < \frac{1}{2}$. Otherwise, 
$$H(P) \geq \log(N)-1 \geq \U(P) -1\;,$$
and we are done.

Let $U$ be the maximal unit set of $P$, i.e $|U|=u=2^{\U(P)}$.
Let $k\in U$ be the element with the highest probability in $U$.
From maximality of $U$ we can assume that $U=\{k,k+1,...,k+u-1\}$.
Calculating,
\[
1\geq\left|\log P(k)-\log P(k+u-1)\right|=(u-1)\log\frac{1}{\alpha}
\]
Therefore, $u=\frac{1}{\log\frac{1}{\alpha}}+1=O(\frac{1}{1-\alpha})$.
To achieve the result it is enough to show that $\frac{1}{1-\alpha}\leq2^{H(P)+O(1)}$,
i.e $H(P)\geq\log\frac{1}{1-\alpha}-O(1)$. Calculating the entropy, indeed,

\[
H(P)=\log\left(\frac{1-\alpha^{N}}{1-\alpha}\right)+\left(\frac{1-N\alpha^{N-1}}{1-\alpha^{N}}\right)\alpha\log\frac{1}{\alpha}+\left(\frac{1-\alpha^{N-1}}{1-\alpha^{N}}\right)\alpha^{2}\cdot\frac{\log\frac{1}{\alpha}}{1-\alpha}\geq\log\left(\frac{1}{1-\alpha}\right)-O(1)\;,
\]
as required.\end{proof}

\begin{claim}
Let $P\in\mathrm{Bin}_{N}$. Then $\U(P)\leq H(P)+O(1)$. \end{claim}
\begin{proof}

Let $p \in (0,1)$ be such that $P(k) = {{N}\choose{k}} p^{k} {(1-p)}^{n-k}$.
Let $U$ be a unit set of $P$ with size $2^{\U(P)}$. We will partition
the codewords in $[N]$ into three regions and bound the number of
codewords from each region in $U$. The regions are:
\begin{enumerate}
\item $\{k \in [N] \mid k > pN + \sqrt{pN} \}$,
\item $\{k \in [N] \mid k <  pN - \sqrt{pN} -1 \}$
\item and  $\{k \in [N] \mid pN - \sqrt{pN} -1 \leq  k \leq pN + \sqrt{pN} \}$.
\end{enumerate}
We will show that in any region,
the number of elements from the region in $U$ is bounded by $O(\sqrt{pN})$.
This will yield a total bound of $|U|=O(\sqrt{pN})$.
The entropy of $P$ is $H(P)=\frac{1}{2}\log\left(2\pi eNp(1-p)\right)$. Therefore  $\sqrt{pN}=2^{H(P)+O(1)}$ and the result follows.

First we consider elements $k$ from the first region. Let $u_{1}$ be the
number of words in $U$ from this region. In this case
\begin{eqnarray*}
\frac{P(k+1)}{P(k)} & = & \frac{{N \choose k+1}p^{k+1}(1-p)^{N-(k+1)}}{{N \choose k}p^{k}(1-p)^{N-k}}=\frac{(N-k)}{(k+1)}\cdot\frac{p}{1-p}\leq\frac{(N-k)}{k}\cdot\frac{p}{1-p}=\left(\frac{N}{k}-1\right)\cdot\frac{p}{1-p}\\
 & \leq & \left(\frac{N}{pN+\sqrt{pN}}-1\right)\cdot\frac{p}{1-p}\leq1-\frac{1}{\sqrt{pN}+1}\;.
\end{eqnarray*}
In a similar way to the proof of Claim~\ref{cla:natural dist: Geometric}, we can conclude that $u_{1}$
is bounded by $O(\sqrt{pN}+1)=O(\sqrt{pN})$

Now consider element $k$ in the second region, similarly:
\begin{eqnarray*}
\frac{P(k+1)}{P(k)} & = & \frac{(N-k)}{(k+1)}\cdot\frac{p}{1-p}\geq\frac{N-(k+1)}{k+1}\cdot\frac{p}{1-p}=\left(\frac{N}{k+1}-1\right)\cdot\frac{p}{1-p}\\
 & \geq & \left(\frac{N}{pN-\sqrt{pN}}-1\right)\cdot\frac{p}{1-p} \geq 1+\frac{1}{\sqrt{pN}}\;.
\end{eqnarray*}
Therefore $u_{2}$, the number of elements from the second region
in $U$, is bounded by $O(\sqrt{pN})$

Clearly ,$u_3$, the number of elements from $U$ in the last region, is bounded by the size of the region. So $u_{3}=O(\sqrt{pn})$.

Combining the above, we get
\[
2^{\U(P)}=|U|=\sum_{i=1}^3 u_{i}=O\left(\sqrt{pn}\right)=2^{H(P)+O(1)}
\]
as required.\end{proof}

}
\subsection{Dependence of communication on entropy}
\label{ssec:linear}

In the previous sections we gave a scheme with performance that is
exponential in the entropy.
This scheme is error-free for some natural distributions and
had positive error for general
distributions.
The next lemma shows that if we cannot find a scheme with performance that are linear in the entropy,
then any scheme that we will find must have positive error for some distributions.

\begin{lemma}
\label{lem:linear}
For every non-decreasing
function $L:\R^+ \to \R^+$ there exists a constant $c = c_L$
such that the following holds:
If there exists $\left(\Delta,L(H(P))\right)$-UCS for some $\Delta > 0$,
then there exists a $\left(\Delta,c\cdot (1+H(P))\right)$-UCS.
\end{lemma}

\begin{proof}
We will prove the lemma for $c = L(3) + 2$.
Let $(E,D)$ be the $(\Delta,L(H(P))$-UCS. We will construct
a UCS $(E',D')$ that has the required performance.

For every distribution $P \in \calp([N])$ and real number $M > 1$,
we introduce a notion of an $M$-concentrated version of $P$, denote
$P_M$, to be: $P_M(1) = 1 - 1/M + (1/M)\cdot P(1)$ and $P_M(i) = (1/M)\cdot P(i)$
for $i > 1$. So $P_M$ is mostly focussed on a single point and so
has small entropy, but it provides enough variability to capture the
variation of $P$. In what follows, we will apply $(E,D)$ to the distributions
$P_M$ and $Q_M$ for an appropriate choice of $M$, chosen to reduce the
entropy of $P_M$ to be a constant and this will give
the schemes $E'$ and $D'$.

\paragraph{The new scheme $(E',D')$:} On input $P \in \calp([N])$ and $m \in [N]$,
$E'(P,m)$ is computed as follows: Let $M = \left\lceil H(P) \right\rceil$.
Then $E'(P,m) = (M,E(P_M,m)).$

On input $Q$ and received string $y' = (M,y)$ the decoding
$D'(Q,y') = D(Q_M,y)$.

In what follows we argue that this is a valid zero-error UCS for
uncertainty parameter $\Delta$, with performance $c \cdot H(P)$.
We start by proving its validity.

\begin{claim}
For every pair $P,Q \in \calp([N])$ such that $\delta(P,Q) \leq \Delta$,
and for every $m \in [N]$ we have $D'(Q,E'(P,m)) = m$.
\end{claim}

\begin{proof}
Fix $M = \lceil H(P) \rceil$. Since $E'(P,m) = E(P_M,m)$
and $D'(Q,(M,y)) = D(Q_M,y)$, it suffices to prove that
$P_M$ and $D_M$ are $\Delta$-close, since then we
can use the correctness of $(E,D)$ on $P_M$ and $Q_M$
to conclude $D(Q_M,E(P_M,m)) = m$. Below we verify that
$P_M$ and $Q_M$ are $\Delta$-close.

First we consider $m \in [N] \setminus \{1\}$.
For such $m$ we have $P_{M}(m)=\frac{1}{M}P(m)$
and $Q_{M}(m)=\frac{1}{M}Q(m)$ and so
$P_{M}(m)/Q_{M}(m) = P(m)/Q(m)$.
So $|\log P_M(m)/Q_M(m)| = |\log P(m)/Q(m)| \leq \Delta$.

Now, consider $m=1$. In this case
$P_{M}(m)=\left(\frac{M-1}{M}\right)+\left(\frac{1}{M}\right)\cdot P(1)$
and $Q_{M}(m)=\left(\frac{M-1}{M}\right)+\left(\frac{1}{M}\right)\cdot Q(1)$.
Assume $P(1) \geq Q(1)$ (the other case is similar)
and so $0 \leq \log P(1)/Q(1) \leq \Delta$.
On the one hand we have $P_M(1) \geq Q_M(1)$
and on the other hand we have $P_M(1)/Q_M(1) \leq P(1)/Q(1)$ (which holds for
every $M > 0$).
It follows that $0 \leq \log P_M(1)/Q_M(1) \leq \log P(1)/Q(1) \leq \Delta$.

It follows that $\delta(P_M,Q_M) \leq \Delta$ and the claim follows.
\end{proof}

It remains to analyze the performance of the scheme.

\begin{claim}
For every distribution $P \in \calp([N])$,
we have
$\E\left[|E'_{m\sim_{P}[N]}(P,m)|\right] \leq c\cdot H(P)$.
\end{claim}

\begin{proof}
Recall that the encoding of $m \in [N]$ is the pair
$(M,E(P_M(m))$ where $M = \lceil H(P) \rceil$. It
follows that the first part the encoding is always
of length at most $2\cdot (1 + \log H(P))$
(allowing for prefix free encodings and rounding up of $H(P)$ to
its ceiling). We crudely bound the above by $2(1 + H(P))$.

We turn to the length of the second part, i.e., $E(P_M(m))$.
We first show that
$\E\left[|E_{m\sim_{P}[N]}(P_M,m)|\right] \leq M \cdot
\E\left[|E_{m\sim_{P_M}[N]}(P_M,m)|\right]$.
We then bound
$\E\left[|E_{m\sim_{P_M}[N]}(P_M,m)|\right]$ by $L(3)$ thus giving us
that total expected length of the encoding $\E_{m \sim_P}[ |E'(P,m) | ] \leq
(L(3) + 2)\cdot (1+H(P)) = c(1+H(P))$.

We start by showing the first step. We have
\begin{eqnarray*}
\lefteqn{\E_{m\sim_{P_{M}}[N]}\left[|E(P_{M},m)|\right]} \\
 & = & \frac{1}{M}\E_{m\sim_{P}[N]}\left[|E(P_{M},m)|\right]+\left(1-\frac{1}{M}\right)\E_{m\sim_{P}[N]}\left[|E(P_{M},1)|\right]\\
 & \geq & \frac{1}{M} \E_{m\sim_{P}[N]}\left[|E(P_{M},m)|\right].
\end{eqnarray*}
It follows that
$\E_{m\sim_{P}[N]}\left[|E(P_M,m)|\right]\leq
M\E_{m\sim_{P_{M}}[N]}\left[|E(P_{M},m)|\right]$ as asserted.

By the performance of $E$ on $P_M$, we have
$\E_{m\sim_{P_{M}}[N]}\left[|E(P_{M},m)|\right] \leq L(H(P_M))$.
So it suffices to show $H(P_M) \leq 3$. This is straightforward from the
definition of $P_M$ and the choice of $M$.
We have
\begin{eqnarray*}
H(P_M) & = & \sum_{m \in [N]} P_M(m) \log 1/P_M(m) \\
& \leq & (1 - 1/M) \log 1/P_M(1) +
(1/M) \cdot \sum_{m \in [N]} P(m) \log M/P(m) \\
& \leq & 1 + 1/M \cdot (H(P) + \log M) \\
& & \mbox{(Using $P_M(1) \geq 1/2$ if $M\geq 2$ and $1-1/M = 0$ otherwise.)} \\
& \leq & 1 + 1 + \log M/M \\
& \leq & 3
\end{eqnarray*}
as required.

The claim follows and so does the lemma.
\end{proof}
\end{proof}

\bibliographystyle{plain}
\bibliography{coding}

\end{document}